\documentclass{article}%
\pdfoutput=1
\usepackage{amsfonts}
\usepackage{amssymb}
\usepackage[nosumlimits,nointlimits,nonamelimits]{amsmath}
\usepackage[onehalfspacing]{setspace}
\usepackage{footmisc}
\usepackage{graphicx}%
\providecommand{\U}[1]{\protect\rule{.1in}{.1in}}
\newtheorem{theorem}{Theorem}[section]

\newtheorem{corollary}{Corollary}[section]

\newtheorem{definition}{Definition}[section]

\newtheorem{lemma}{Lemma}[section]

\newtheorem{proposition}{Proposition}[section]
\newtheorem{remark}{Remark}[section]

\newenvironment{proof}[1][Proof]{\noindent\textbf{#1.} }{\ \rule{0.5em}{0.5em}}
\begin{document}

\title{A Classification of Hidden-Variable Properties\thanks{We are indebted to
Shelly Goldstein for detailed and valuable comments on an earlier
draft. \ Our thanks to Amanda Friedenberg, Gus Stuart, and three anonymous referees for important input,
and also to Michael Giangrasso, Peter Morgan, Alex Peysakhovich, and Jean de Valpine.\ \
Financial support from the Stern School of Business is gratefully
acknowledged. \ {\tiny chvp-12-03-08}}}
\author{Adam Brandenburger\thanks{Address: Stern School of Business, New York
University, 44 West Fourth Street, New York, NY 10012,
adam.brandenburger@stern.nyu.edu, www.stern.nyu.edu/$\sim$abranden}
\and Noson Yanofsky\thanks{Address: Department of Computer and Information Science,
Brooklyn College, 2900 Bedford Avenue, Brooklyn, NY 11210,
noson@sci.brooklyn.cuny.edu, www.sci.brooklyn.cuny.edu/$\sim$noson}}
\date{First Version 01/04/07\\
Publication Version 09/26/08} \maketitle

\begin{abstract}
Hidden variables are extra components
added to try to banish counterintuitive features of quantum mechanics. We
start with a quantum-mechanical model and describe various properties that can
be asked of a hidden-variable model. We
present six such properties and a Venn diagram of how they are
related. With two existence theorems and three no-go theorems (EPR, Bell,
and Kochen-Specker), we show which properties of empirically equivalent
hidden-variable models are possible and which are not. Formally, our
treatment relies only on classical probability models, and physical
phenomena are used only to motivate which models to choose.
\end{abstract}

\section{Introduction}

\thispagestyle{empty}Begun by von Neumann \cite[1932]{vonneumann32}, the
hidden-variable program in quantum mechanics (QM) adds extra
(\textquotedblleft hidden\textquotedblright) ingredients in order to try to
banish some of the counterintuitive features of QM. \ These features are: (i)
the probabilistic nature of quantum behavior, (ii) the possibility of
so-called non-local effects between widely separated particles, and (iii) the
idea of an intrinsic dependence between the observer of a QM system and the
properties of the system itself.

Hidden-variable theories aim to remove these strange aspects of QM by building
more \textquotedblleft complete\textquotedblright\ models (in the terminology
of Einstein-Podolsky-Rosen \cite[1935]{einstein-podolsky-rosen35}). \ The
completed models should agree with the predictions of QM, but exhibit one or
more of the desired properties of: (i) determinism, (ii) locality, and (iii)
independence.

Can such models actually be built? \ The famous \textquotedblleft
no-go\textquotedblright\ theorems of QM show that there are severe limitations
to what can be done. \ But it is also true that certain combinations of
properties are possible.

Our modest goal in this paper is to provide a formal framework in
which various properties one might ask of hidden-variable models can
be stated and in which various non-existence and existence results
can be organized. \ Almost all--if not all--of the ingredients of
what we do in this paper are well known to researchers in the area.
\ Our contribution, we hope, is in
putting all the ingredients into one simple setting.

The setting is classical probability spaces. The question is, given
a classical probability model, whether there exists an associated
hidden-variable model that is empirically equivalent to the first model
and that satisfies certain properties. These properties are
motivated by the literature on hidden variables in QM. The specific
properties we consider--and the relationships among them--can be
depicted in the Venn diagram of Figure 1.1. (We define all the terms
later.) The diagram contains 21 regions.

\[%
{\parbox[b]{5.5936in}{\begin{center}
\includegraphics
[trim=0.000000in 0.000000in -0.381536in 0.000000in,
natheight=5.010700in,
natwidth=8.083400in,
height=3.3217in,
width=5.5936in
]{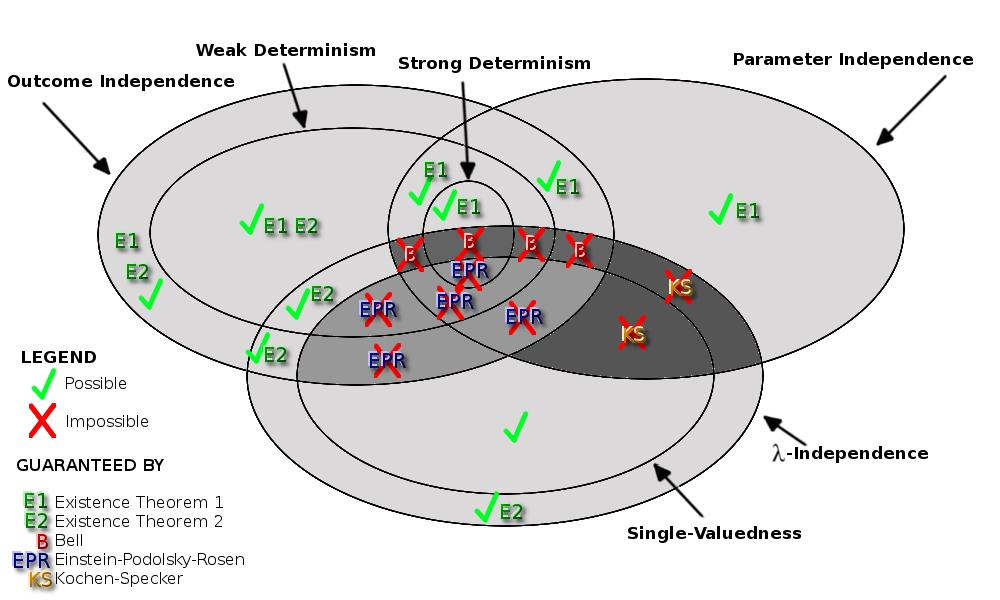}%
\\
Figure 1.1
\end{center}}}%
\]

The main result of the paper is that we can give a complete account
of these 21 regions. For 10 of these regions (indicated with
checks), it is always possible to find an equivalent hidden-variable
model with the properties in question. For the remaining 11 regions
(indicated with crosses), this may not be possible. We fill in the
regions via two existence results and three non-existence results.
The latter three are the famous theorems of Einstein-Podolsky-Rosen
(EPR) \cite[1935]{einstein-podolsky-rosen35},\footnote{Strictly
speaking, EPR did not state a non-existence theorem, but it is
useful to present their argument this way.} Bell
\cite[1964]{bell64}, and Kochen-Specker
\cite[1967]{kochen-specker67}.

It is important to understand that, formally, our paper makes no use
of physical phenomena. It is an exercise in classical probability
theory alone. Of course, the probability spaces we select for the
non-existence results are inspired by the physical experiments
described in EPR, Bell, and Kocher-Specker. But we hope it is
conceptually clarifying to present the hidden-variable question in
a purely abstract setting--that is, to show how much follows from
the rules of probability theory alone.

Naturally, our account of hidden-variable theory is complete only
relative to the properties we consider (there are six of them).
These properties are, as far as we can tell, the main ones
considered in the literature. (As we explain later, we have added
one definition.) In particular, Bell Locality (\cite[1964]{bell64})
is equivalent to the conjunction of Outcome and Parameter
Independence (Jarrett \cite[1984]{jarrett84}, stated here as
Proposition 2.1). Kochen-Specker \cite[1967]{kochen-specker67}
Non-Contextuality is implied by the conjunction of Parameter
Independence and $\lambda$-Independence (Proposition 2.2). But there
may well be other interesting properties to put on hidden-variable
models, which would lead to an extension of Figure 1.1.

There are, of course, many treatments of the hidden-variable
question. (Prominent examples include Belinfante
\cite[1973]{belinfante}, Gudder \cite[1988]{gudder88},
Mermin \cite[1993]{mermin93}, Peres
\cite[1990]{peres90}, and van Fraassen
\cite[1991]{vanfraassen}.) This paper is not meant to be a
comprehensive survey. Our goal is, in a sense, the reverse. It is to
start with the rules of probability theory alone and ask--relative
to the six properties we consider--what is or is not possible. For
this task, we need only EPR, Bell, and Kochen-Specker. But we do
mention later some other no-go results in QM (not needed to complete
Figure 1.1).

Two comments on the particular framework we present.  First, we work
with a single probability measure on a single space, where points in
the space describe measurements on particles and outcomes of those
measurements.  An alternative--more conventional--approach would be to
use a family of probability measures on a space describing outcomes
only, with different probability measures corresponding to different
measurements.\footnote{We are grateful to Shelly Goldstein and a
referee who both pointed to this issue.}  In fact, all our
requirements are stated in terms of conditional probabilities: If
such-and-such a measurement is made, then what is the probability of
a certain outcome?  The distinction between the approaches might
therefore seem small--both involve families of probability measures.
But it matters.  If we had started from a family of probability
measures rather than a single measure, we would not have been able
to derive all the relationships between properties shown in Figure
1.1 without making some additional assumptions.\footnote{A reader
preferring the more conventional approach can simply make these
additional assumptions. We give details in the next section.}  So,
formally, our approach is more parsimonious.  Yet, it does add an
ingredient at the conceptual level--viz., the existence of a
probability measure prior to conditioning on measurements.  This
measure may be thought of as representing the perspective of a
"super-observer" who observes the experimenters as well as the
outcomes of the experiments.  Does the existence of such a measure
contradict the free will of an experimenter in deciding what
measurements to make? We don't think so--because, as we said, we work
only with the conditionals.  Still, even if it plays a very small
role in our treatment, the idea of such a measure seems deserving of
further study.  We leave this as beyond the scope of the current
paper.

A second choice we make in our framework is to treat only finite
probability spaces.  This involves a tradeoff.  On the one hand,
finiteness allows us to avoid all measure-theoretic issues. On the
other hand, as an assumption on the space in which a hidden variable
lives, finiteness is undoubtedly restrictive. To be precise, the
first of our two existence theorems needs only a finite space in any
case, but, under finiteness, the second can treat only rational
probabilities. (We sketch the extension of our second theorem, using
an infinite space, to all probabilities.) Of course, finiteness
makes our versions of the no-go theorems weaker.

We derive Figure 1.1 in the body of the paper. Before that, though,
let us offer a comment on its conceptual meaning in QM. The main
message of the no-go theorems is that in building a hidden-variable
theory, some properties that might be viewed as desirable--at least,
a priori--have to be given up. But there is a choice of what to give
up. Arguably, it is more a matter of metaphysics than physics as to
what choice to make. The point of a formal treatment--as in Figure
1.1--is to give a precise statement of what the options are. There
is a basic three-way tradeoff. We can have:

\begin{enumerate}
\item[(i)] \textbf{Determinism} \ (As we will explain, this comes in a strong
or a weak form.) \ This says that randomness reflects only observer ignorance.
\ Once hidden variables are introduced, there is no residual randomness in the universe.

\item[(ii)] \textbf{Parameter Independence} \ This says that when conducting
an experiment on a system of particles, the outcome of a measurement on one of
the particles does not depend on what measurements are performed on other
particles. \ (The intuitive appeal of this property is that often the
particles are widely separated.) \ This is a way of saying that the universe
is local.

\item[(iii)] \textbf{$\lambda$-Independence} \ This says that the nature of the
particles--as determined by the value of a hidden variable--does not depend on
the experiment conducted. \ There is, in this sense, no dependence between the observer and
the observed.
\end{enumerate}

Any one of these properties is consistent with the predictions of QM. \ So are
certain combinations of properties, as Figure 1.1 shows. \ But, however a
priori reasonable they may seem, we cannot have all three properties--or even
certain pairs of properties. \ There is an inherent tradeoff. \ This is an
inescapable feature of QM.

The rest of the paper is organized as follows. \ Section 2 lays out the framework
and basic definitions. \ Section 3 presents two existence theorems on
hidden-variable models. \ Sections 4-6 present EPR \cite[1935]{einstein-podolsky-rosen35},
Bell \cite[1964]{bell64}, and Kochen-Specker
\cite[1967]{kochen-specker67} in our probability-theoretic framework. Section 7 mentions some other impossibility
results in QM not covered in this paper.

\section{The Models and Their Properties}

Here is the hidden-variable question in a bit more detail. \ Start with a
model of an experiment done in QM. \ Sometimes, the experiment will consist of
measurements performed on several entangled\ particles. \ Or, the experiment
might involve several measurements performed on a single particle. \ The model
describes the set-up and outcome of the experiment and so will be called an
\textquotedblleft empirical model.\textquotedblright\ \ The question is
whether one can find a hidden-variable model--i.e., a model involving
additional (\textquotedblleft hidden\textquotedblright) variables--which is
empirically equivalent to the first model and which has desired properties.
\ By \textquotedblleft empirically equivalent\textquotedblright\ we mean that
the two models make the same (probabilistic) prediction about outcomes.

Formally, we consider a space%
\[
\Psi=\{a,a^{\prime},\ldots\}\times\{b,b^{\prime},\ldots\}\times\{c,c^{\prime
},\ldots\}\cdots\times\{A,A^{\prime},\ldots\}\times\{B,B^{\prime}%
,\ldots\}\times\{C,C^{\prime},\ldots\}\times\cdots.
\]
The variables $A,B,C$, \ldots\ are measurements, and the variables $a,b,c$,
\ldots\ are associated outcomes of measurements. \ There might be several
particles:\ Ann performs a measurement on her particle, Bob performs a
measurements on his particle, \ldots. \ Or, $\Psi$ might describe a case where
several measurements are performed on one particle. \ The definitions to come
apply in either case. \ We take each of the spaces in $\Psi$ to be finite, and
suppose that $\Psi$ is a finite product.

Let $\Lambda$ be a finite space in which a hidden variable $\lambda$
lives.\footnote{Throughout, we talk about one hidden variable. \ But
we put no structure on the space in which the hidden variable lives.
Of course, in the infinite case, a measurable structure would be
needed.} \ The overall space is then%
\[
\Omega=\Psi\times\Lambda.
\]

\begin{definition}
An \textbf{empirical model} is a pair $(\Psi,q)$, where $q$ is a probability
measure on $\Psi$. \ A \textbf{hidden-variable\ model} is a pair $(\Omega,p)$,
where $p$ is a probability measure on $\Omega$.
\end{definition}

Essentially, we employ the probability measures $q$ and $p$ once
conditioned on one or more measurements.  Still, as we said in the
Introduction, we cannot quite dispense with the unconditional $q$
and $p$, and work instead with a family of probability measures
indexed by measurements (one family for $q$ and one for $p$).  If we
did, we would lose Lemmas 2.1 and 2.4--and hence certain
relationships in Figure 1.1.  These relationships seem intuitively
correct to us, so we prefer a formalism in which they can be
derived.  With an indexed family of measures, we would have to
impose rather than derive these relationships.

The models of Definition 2.1 stand alone. However, we are also
interested in stating when different models are equivalent:

\begin{definition}
An empirical model $(\Psi,q)$ and a hidden-variable model $(\Omega,p)$ are
(\textbf{empirically})\textbf{\ equivalent} if for all $a,b,c$, \ldots, $A,B,C
$, \ldots,%
\[
q(A,B,C,\ldots)>0\text{ if and only if }p(A,B,C,\ldots)>0,
\]
and when both are non-zero,%
\[
q(a,b,c,\ldots|A,B,C,\ldots)=p(a,b,c,\ldots|A,B,C,\ldots).
\]

\end{definition}

Here,\textbf{\ }we write \textquotedblleft$a,b,c,\ldots$\textquotedblright\ as
a shorthand for the event%
\[
\{(a,b,c,\ldots)\}\times\{A,A^{\prime},\ldots\}\times\{B,B^{\prime}%
,\ldots\}\times\{C,C^{\prime},\ldots\}\times\cdots
\]
in $\Psi$, or the event%

\[
\{(a,b,c,\ldots)\}\times\{A,A^{\prime},\ldots\}\times\{B,B^{\prime}%
,\ldots\}\times\{C,C^{\prime},\ldots\}\times\cdots\times\Lambda
\]
in $\Omega$, and similarly for other expressions. \ We will adopt this
shorthand throughout.

The non-nullness condition is simply to ensure that any measurements
$(A,B,C,\ldots)$ which are possible in the empirical model are also
possible in the hidden-variable model under consideration, and vice
versa. \ Without this condition, it would be hard to compare the two
models.

We will often calculate $p(a,b,c,\ldots|A,B,C,\ldots)$, for $p(A,B,C,\ldots
)>0$, from the formula:%
\[
p(a,b,c,\ldots|A,B,C,\ldots)=\sum\limits_{\{\lambda:p(A,B,C,\ldots
,\lambda)>0\}}p(a,b,c,\ldots|A,B,C,\ldots,\lambda)p(\lambda|A,B,C,\ldots).
\]
Substituting this in Definition 2.2, we see that the idea of equivalence is to
reproduce a given probability measure $q$ on the space $\Psi$ by
averaging under a probability measure $p$ on an augmented space $\Omega$,
where $\Omega$ includes a hidden variable. \ The measure $p$ is then subject
to various conditions (described below).

One more basic definition:

\begin{definition}
Two hidden-variable models $(\Omega,p)$ and $(\Omega,\overline{p})$ are
(\textbf{empirically})\textbf{\ equivalent} if for all $a,b,c$, \ldots, $A,B,C
$, \ldots,%
\[
p(A,B,C,\ldots)>0\text{ if and only if }\overline{p}(A,B,C,\ldots)>0,
\]
and when both are non-zero,%
\[
p(a,b,c,\ldots|A,B,C,\ldots)=\overline{p}(a,b,c,\ldots|A,B,C,\ldots).
\]

\end{definition}

Now, we move on to the different properties of hidden-variable models.
\ Figure 2.1 repeats Figure 1.1 (without the checks and crosses), as a preview
of the properties and relationships we will consider.%
\[%
{\parbox[b]{5.5936in}{\begin{center}
\includegraphics[
trim=0.000000in 0.000000in -0.381536in 0.000000in,
natheight=5.010700in,
natwidth=8.083400in,
height=3.3217in,
width=5.5936in
]%
{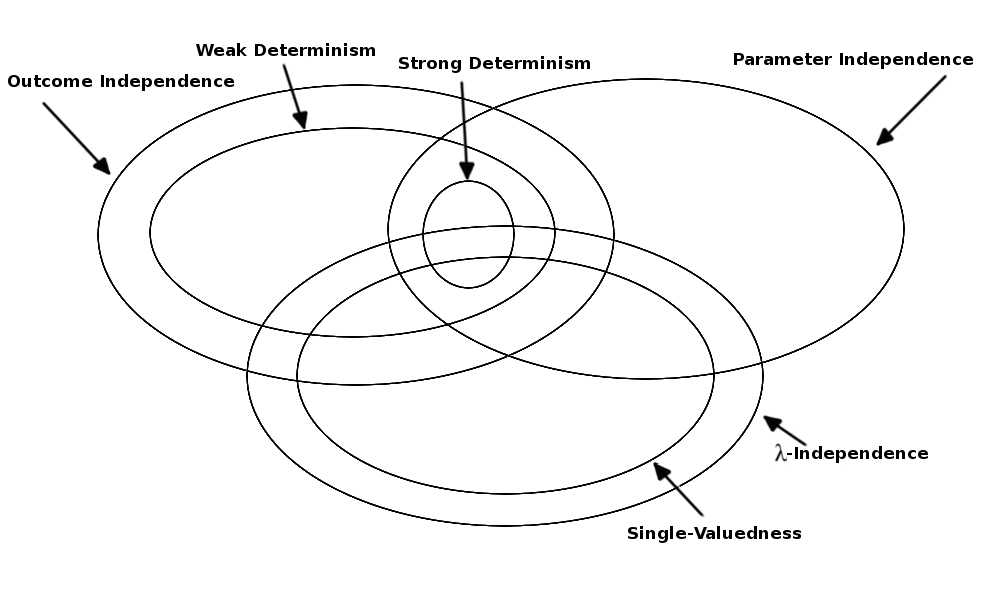}%
\\
Figure 2.1
\end{center}}}%
\]

\begin{definition}
A hidden-variable model $(\Omega,p)$ satisfies \textbf{Single-Valuedness} if
$\Lambda$ is a singleton.
\end{definition}

This condition says that the hidden variable can take on only one value. \ In
effect, this condition doesn't allow hidden variables. \ We include it because
EPR will be usefully formulated this way.

\begin{definition}
A hidden-variable model $(\Omega,p)$ satisfies
\textbf{$\lambda$-Independence} if for all $A,A',B,B',C,C'$, $\ldots$,
$\lambda$, whenever
\[
p(A,B,C, \ldots)> 0\mbox{ and } p(A',B',C', \ldots) > 0,
\]
then
\[
p(\lambda|A,B,C,\ldots)=p(\lambda |A',B',C',\ldots).
\]
\end{definition}

(This term is from Dickson \cite[2005, p.140]{dickson}.)
The condition says that the process determining the value of the hidden variable is independent of which measurements are chosen.

\begin{remark}
If a hidden-variable model satisfies Single-Valuedness, then it satisfies $\lambda$-Independence.
\end{remark}

\begin{definition}
A hidden-variable model $(\Omega,p)$ satisfies \textbf{Strong Determinism} if,
for every $A,\lambda$, whenever $p(A,\lambda)>0$, there is an $a$ such that
$p(a|A,\lambda)=$ $1$, and similarly for $B,\lambda,b$, etc.
\end{definition}

\begin{definition}
A hidden-variable model $(\Omega,p)$ satisfies \textbf{Weak Determinism} if,
for every $A,B,C$, $\ldots,$ $\lambda$, whenever $p(A,B,C,\ldots,\lambda)>0$,
there is a tuple $a,b,c$, \ldots\ such that $p(a,b,c,\ldots|A,B,C,\ldots
,\lambda)=$ $1$.
\end{definition}

Determinism is a basic condition in the literature. \ But we are careful to
make a distinction between a strong and weak form. \ We will see that various
results are true for one form but false for another. \ Broadly, the condition
is that the hidden variable determines (almost surely) the outcomes of
measurements. \ But, Strong Determinism says this holds
measurement-by-measurement, while Weak Determinism says this holds only once
all measurements are specified. \ There is a one-way implication:

\begin{lemma}
If a hidden-variable model satisfies Strong Determinism, then it satisfies
Weak Determinism.
\end{lemma}

\begin{proof}
Suppose $p(A,B,C,\ldots,\lambda)>0$. \ Then $p(A,\lambda)>0,p(B,\lambda
)>0,p(C,\lambda)>0$, \ldots. \ So, there are $a,b,c$, \ldots, such that
$p(a|A,\lambda)=$ $1,p(b|B,\lambda)=$ $1,p(c|C,\lambda)=$ $1$, \ldots.

The result now follows from the following easy fact in probability
theory:\ Let $E_{1},\ldots,E_{n}$ and $F_{1},\ldots,F_{n}$ be events with
$p(\bigcap_{i}F_{i})>0$. \ If $p(E_{i}|F_{i})=1$ for all $i$, then
$p(\bigcap_{i}E_{i}|\bigcap_{i}F_{i})=1$.
\end{proof}

\begin{definition}
A hidden-variable model $(\Omega,p)$ satisfies \textbf{Outcome Independence}
if for all $a,b,c$, \ldots, $A,B,C$,\ldots, $\lambda$, whenever
$p(A,B,C,\ldots,b,c,\ldots,\lambda)>0$,%
\begin{equation}
p(a|A,B,C,\ldots,b,c,\ldots,\lambda)=p(a|A,B,C,\ldots,\lambda), \tag{2.1}%
\end{equation}
and similarly with $a$ and $b$ interchanged, etc.
\end{definition}

Outcome Independence is taken from Jarrett \cite[1984]{jarrett84} and Shimony
\cite[1986]{shimony86}. \ It says that conditional on the value of the hidden
variable and the measurements undertaken, the outcome of any one measurement
is (probabilistically) unaffected by the outcomes of the other measurements.

\begin{lemma}
A hidden-variable model $(\Omega,p)$ satisfies Outcome Independence if and
only if for all $a,b,c$, \ldots, $A,B,C$,\ldots, $\lambda$, whenever
$p(A,B,C,\ldots,\lambda)>0$,%
\begin{equation}
p(a,b,c,\ldots|A,B,C,\ldots,\lambda)=p(a|A,B,C,\ldots,\lambda)\times
p(b|A,B,C,\ldots,\lambda)\times p(c|A,B,C,\ldots,\lambda)\times\cdots.
\tag{2.2}%
\end{equation}

\end{lemma}

\begin{proof}
Standard; see, e.g., Chung \cite[1974, Theorem 9.2.1]{chung74}.
\end{proof}

\begin{lemma}
[{Bub \cite[1997, p.69]{bub97}}]If a hidden-variable model satisfies Weak
Determinism, then it satisfies Outcome Independence.
\end{lemma}

\begin{proof}
Suppose $p(A,B,C, \ldots, \lambda)>0$. Then, by Weak Determinism, there is a tuple $a^*, b^*, c^*,
\ldots$ such that

$$
p(a,b,c, \ldots | A,B,C, \ldots, \lambda) = \chi_{\{a=a^*\}} \times \chi_{\{b=b^*\}} \times \chi_{\{c=c^*\}} \times \cdots .
$$

But then $p(a | A,B,C,\ldots, \lambda)= \chi_{\{a=a^*\}}$, $p(b | A,B,C,\ldots, \lambda)= \chi_{\{b=b^*\}}$,
$p(c | A,B,C,\ldots, \lambda)= \chi_{\{c=c^*\}}, \ldots .$ Now use Lemma 2.2.
\end{proof}

\begin{definition}
A hidden-variable model $(\Omega,p)$ satisfies \textbf{Parameter Independence}
if for all $a,A$,$B,C$, \ldots, $\lambda$, whenever $p(A,B,C,\ldots,\lambda)>0$,%
\begin{equation}
p(a|A,B,C,\ldots,\lambda)=p(a|A,\lambda), \tag{2.3}%
\end{equation}
and similarly for $b,A,B,C$, \ldots, $\lambda$, etc.
\end{definition}

Parameter Independence is also from Jarrett \cite[1984]{jarrett84} and Shimony
\cite[1986]{shimony86}. \ It says that, conditional on the value of the hidden
variable, the outcome of any one measurement depends (probabilistically) only
on that measurement and not on the other measurements.

\begin{lemma}
If a hidden-variable model satisfies Strong Determinism, then it satisfies
Parameter Independence.
\end{lemma}

\begin{proof}
Suppose $p(A,B,C,\ldots,\lambda)>0$. \ Then $p(A,\lambda)>0$. \ So, by Strong
Determinism, there is an $a^{\ast}$ such that $p(a|A,\lambda)=\chi
_{\{a=a^{\ast}\}}$. \ But $p(a,A,B,C,\ldots,\lambda|A,\lambda
)=p(a|A,B,C,\ldots,\lambda)\times p(A,B,C,\ldots,\lambda|A,\lambda)$, where
$p(A,B,C,\ldots,\lambda|A,\lambda)>0$. \ From $p(a^{\ast}|A,\lambda)=1$,
$p(a^{\ast},A,B,C,\ldots,\lambda|A,\lambda)$ $=p(A,B,C,\ldots,\lambda|A,\lambda
)$. \ From $p(a|A,\lambda)=0$ when $a\neq a^{\ast}$, we get $p(a,A,B,C,\ldots
,\lambda|A,\lambda)=0$. \ Thus, $p(a|A,B,C,\ldots,\lambda)=\chi_{\{a=a^{\ast
}\}}$, establishing (2.3).
\end{proof}%

\medskip
Combinations of some of these properties give the well-known properties of
Locality and Non-Contextuality. \ First is Locality, formulated by Bell
\cite[1964]{bell64}. \ In words, a hidden-variable model satisfies Locality if
the probability of getting some tuple of outcomes factorizes under the measurements.

\begin{definition}
A hidden-variable model $(\Omega,p)$ satisfies \textbf{Locality} if for all
$a,b,c$, \ldots, $A,B,C$, ,\ldots, $\lambda$, whenever $p(A,B,C,\ldots
,\lambda)>0$,%
\begin{equation}
p(a,b,c,\ldots|A,B,C,\ldots,\lambda)=p(a|A,\lambda)\times p(b|B,\lambda)\times
p(c|C,\lambda)\times\cdots. \tag{2.4}%
\end{equation}

\end{definition}

\begin{proposition}
[{Jarrett \cite[1984, p.582]{jarrett84}}]A hidden-variable model satisfies
Locality if and only if it satisfies Outcome Independence and Parameter Independence.
\end{proposition}

\begin{proof}
Assume $p(A,B,C,\ldots,\lambda)>0$, and substitute (2.3) and its counterparts
into (2.2). \ This yields (2.4).

Conversely, assume again $p(A,B,C,\ldots,\lambda)>0$, and sum both sides of
(2.4) over $b,c$, \ldots. \ This yields (2.3). \ Moreover, substituting (2.3)
and its counterparts into (2.4) yields (2.2).
\end{proof}%

\medskip
Non-Contextuality, due to Kochen-Specker \cite[1967]{kochen-specker67}, is a
property of an empirical model. \ It says that the probability of obtaining a
particular outcome of a measurement does not depend on the other measurements performed.

\begin{definition}
An empirical model $(\Psi,q)$ satisfies \textbf{Non-Contextuality} if for all
$a,A,B,B^{\prime},C,C^{\prime}$, \ldots, whenever $q(A,B,C,\ldots)>0$ and
$q(A,B^{\prime},C^{\prime},\ldots)>0$,%
\[
q(a|A,B,C,\ldots)=q(a|A,B^{\prime},C^{\prime},\ldots).
\]
Also, the corresponding conditions must hold for $b,A,A^{\prime}%
,B,C,C^{\prime},\ldots$, etc.
\end{definition}

\begin{proposition}
If a hidden-variable model $(\Omega,p)$ satisfies $\lambda$-Independence and Parameter
Independence, then any equivalent empirical model $(\Psi,q)$ satisfies Non-Contextuality.
\end{proposition}

\begin{proof}
We can assume $p(A,B,C,\ldots)>0$ and $p(A,B^{\prime},C^{\prime},\ldots)>0$.
\ Then%
\begin{multline*}
p(a|A,B,C,\ldots)=\sum\limits_{\{\lambda:p(A,B,C,\ldots,\lambda)>0\}}%
p(a|A,B,C,\ldots,\lambda)p(\lambda|A,B,C,\ldots)=\\
\sum\limits_{\{\lambda:p(A,B,C,\ldots,\lambda)>0\}}p(a|A,B,C,\ldots
,\lambda)p(\lambda)=\\
\sum\limits_{\{\lambda:p(A,B,C,\ldots,\lambda)>0\}}p(a|A,\lambda)p(\lambda),
\end{multline*}
where the second line uses $\lambda$-Independence and the third line uses Parameter
Independence. \ Using $p(A,B,C,\ldots)>0$ and $\lambda$-Independence again, we have
$p(A,B,C,\ldots,\lambda)>0$ if and only if $p(\lambda)>0$. \ So%
\[
p(a|A,B,C,\ldots)=\sum\limits_{\{\lambda:p(\lambda)>0\}}p(a|A,\lambda
)p(\lambda).
\]

A similar argument establishes%
\[
p(a|A,B^{\prime},C^{\prime},\ldots)=\sum\limits_{\{\lambda:p(\lambda
)>0\}}p(a|A,\lambda)p(\lambda),
\]
so that $p(a|A,B,C,\ldots)=p(a|A,B^{\prime},C^{\prime},\ldots)$, as required.
\end{proof}

\section{Two Existence Theorems}

We next prove two existence theorems for hidden-variable models which say what
type of properties can always be found:

\begin{enumerate}
\item[E1] \textit{Given any empirical model, there is an equivalent
hidden-variable model which satisfies Strong Determinism.}

\item[E2] \textit{Given any empirical model, there is an equivalent
hidden-variable model which satisfies Weak Determinism and $\lambda$-Independence.}
\end{enumerate}

That is, each of these sets of conditions on a hidden-variable theory can
always be satisfied. \ They cannot be impeded by any no-go theorems. \ Figure
3.1 repeats part of Figure 1.1, putting a check in a region where there is
always an equivalent hidden-variable model with the properties that hold in
that region. \ The checks are followed by E1 and/or E2 which say which
existence theorem pertains to that region.

(The region for Single-Valuedness alone also has a check. \ The existence of
an equivalent hidden-variable model satisfying Single-Valuedness alone is
immediate--it is essentially just the given empirical model. \ See Remark 3.1
for a statement.)%
\[%
{\parbox[b]{5.5936in}{\begin{center}
\includegraphics[
trim=0.000000in 0.000000in -0.381536in 0.000000in,
natheight=5.010700in,
natwidth=8.083400in,
height=3.3217in,
width=5.5936in
]%
{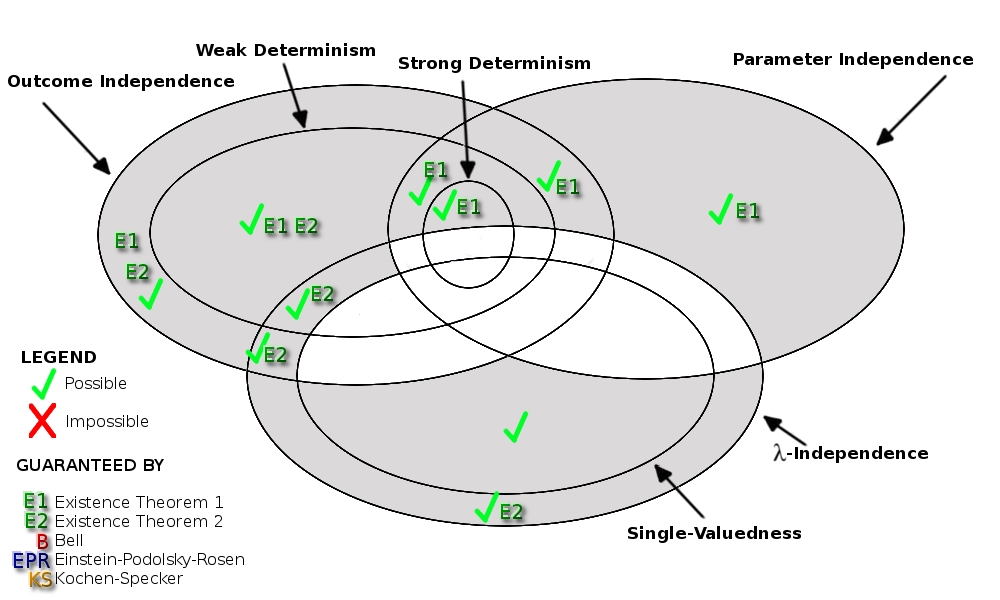}%
\\
Figure 3.1
\end{center}}}%
\]

Here are the two existence theorems. \ Similar methods to those in
the first proof can be found in Fine \cite[1982, p.292]{fine82a}. \
The idea of the second theorem is in
Teufel-Berndl-D\"{u}rr-Goldstein-Zangh\`{\i} \cite[1997,
p.1219]{teufel-berndl-durr-goldstein-zanghi97} (see also Werner and
Wolf \cite[2001, p.7]{werner-wolf01}). \ But we have not found exact
statements of the two theorems in the literature.

\begin{theorem}
Given an empirical model $(\Psi,q)$, there is an equivalent hidden-variable
model $(\Omega,p)$ which satisfies Strong Determinism.
\end{theorem}

The proof is basically a mathematical trick. \ We simply take the hidden
variable to be all the information possible. \ This means, in particular, that
the hidden variable would have to `know' the probabilities for different
measurements and outcomes. \ With this huge hidden variable, we can build up
the probability measure $p$ from the given measure $q$. \ This construction is
physically unsatisfying, of course--but not ruled out by the general concept
of a hidden variable. \ It is also rather obvious. \ Bell \cite[1971]{bell71}
wrote: \textquotedblleft If no restrictions whatever are imposed on the hidden
variables, or on the dispersion-free states, it is trivially clear that such
schemes can be found to account for any experimental results
whatever\textquotedblright\ (reprinted in \cite[p.33]{bell87}). \ Still, we
give a proof--which, in particular, makes clear that even Strong Determinism
is achieved. \ (This won't be possible in the next existence result.)%
\medskip

\begin{proof}
We give the proof for the case that $\Psi$ is a 4-way product, but the
extension to a general (finite) product will be clear. \ Set
\[
\Lambda=\{a,a^{\prime},\ldots\}\times\{b,b^{\prime},\ldots\}\times
\{A,A^{\prime},\ldots\}\times\{B,B^{\prime},\ldots\},
\]
and define $p$ in stages, as follows. \ (Figure 3.2 shows the construction.)
\ For any pair $A,B$, set%
\begin{equation}
p(A,B)=q(A,B). \tag{3.1}%
\end{equation}

For any pair $A,B$, and $\lambda=(\tilde{a},\tilde{b},\tilde{A},\tilde{B})$,
set%
\begin{equation}
p(\lambda|A,B)=\left\{
\begin{array}
[c]{ll}%
q(\tilde{a},\tilde{b}|A,B) & \text{if }\tilde{A}=A\text{ and }\tilde
{B}=B\text{,}\\
0 & \text{otherwise.}%
\end{array}
\right.  \tag{3.2}%
\end{equation}

For pairs $a,b$ and $A,B$, and $\lambda=(\tilde{a},\tilde{b},\tilde{A}%
,\tilde{B})$, set%
\begin{equation}
p(a,b|A,B,\lambda)=\left\{
\begin{array}
[c]{ll}%
1 & \text{if }\tilde{a}=a\text{, }\tilde{b}=b\text{, }\tilde{A}=A\text{,
}\tilde{B}=B\text{,}\\
0 & \text{otherwise.}%
\end{array}
\right.  \tag{3.3}%
\end{equation}

This defines a measure $p$ on $\Omega$ using%
\[
p(a,b,A,B,\lambda)=p(a,b|A,B,\lambda)\times p(\lambda|A,B)\times p(A,B).
\]
(Note that $p(\cdot,\cdot|A,B,\lambda)$ is not a measure if $A\neq\tilde{A}$
or $B\neq\tilde{B}$. \ But then $p(\lambda|A,B)=0$, so there is no difficulty.)

From (3.1), $p(A,B)>0$ if and only if $q(A,B)>0$. \ If both are positive, then
from Figure 3.2,
\[
p(a,b|A,B)=1\times q(a,b|A,B),
\]
so that equivalence is satisfied.

\begin{center}
\includegraphics[
natheight=7.499600in,
natwidth=9.999800in,
height=3.4042in,
width=4.8000in
]
{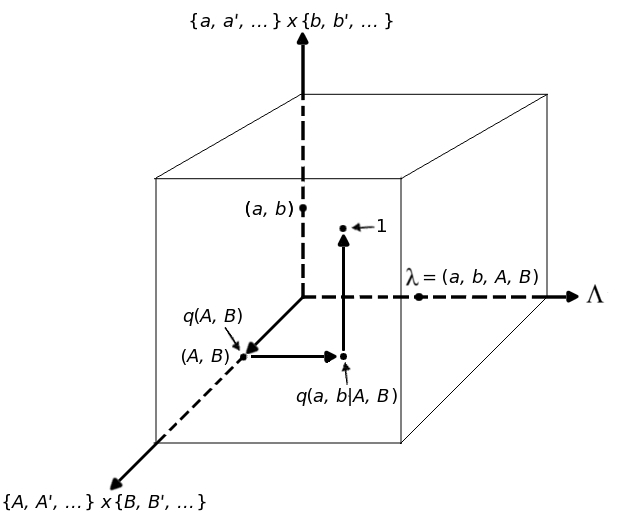}
\\
Figure 3.2
\end{center}

It remains to verify that $(\Omega,p)$ satisfies Strong Determinism. \ So,
suppose $p(A,\lambda)>0$. \ Writing $\lambda=(\tilde{a},\tilde{b},\tilde
{A},\tilde{B})$, we therefore assume $A=\tilde{A}$. \ Using (3.1)-(3.3),%
\begin{align*}
p(a,A,\lambda)  &  =p(a,\tilde{A},\lambda)=\sum\limits_{b^{\prime},B^{\prime}%
}p(a,b^{\prime},\tilde{A},B^{\prime},\lambda)=p(a,\tilde{b},\tilde{A}%
,\tilde{B},\lambda)=q(a,\tilde{b},\tilde{A},\tilde{B})\times\chi
_{\{a=\tilde{a}\}},\\
p(A,\lambda)  &  =p(\tilde{A},\lambda)=\sum\limits_{a^{\prime},b^{\prime
},B^{\prime}}p(a^{\prime},b^{\prime},\tilde{A},B^{\prime},\lambda)=p(\tilde
{a},\tilde{b},\tilde{A},\tilde{B},\lambda)=q(\tilde{a},\tilde{b},\tilde
{A},\tilde{B}),
\end{align*}
so that%
\[
p(a|A,\lambda)=\chi_{\{a=\tilde{a}\}},
\]
which is Strong Determinism.
\end{proof}

\begin{corollary}
Given a hidden-variable model $(\Omega,p)$, there is an equivalent
hidden-variable model $(\Omega,\overline{p})$ which satisfies Strong Determinism.
\end{corollary}

\begin{proof}
Start with $(\Omega,p)$, and (partially) define an equivalent empirical model
$(\Psi,q)$ by%
\[
q(a,b|A,B)=\sum\limits_{\{\lambda:p(A,B,\lambda)>0\}}p(a,b|A,B,\lambda
)p(\lambda|A,B),
\]
for $p(A,B)>0$. \ (There is no difficulty in completing the definition of $q$.)

By Theorem 3.1, there is a hidden-variable model $(\Omega,\overline{p})$ which
is equivalent to $(\Psi,q)$ and which satisfies Strong Determinism. \ But
$(\Omega,\overline{p})$ is also equivalent to $(\Omega,p)$.
\end{proof}%

\medskip
We state the next existence result for the case of rational probabilities, to
enable us to prove it with a finite set $\Lambda$. \ After the proof, we
sketch the extension to all probabilities.

\begin{theorem}
Given an empirical model $(\Psi,q)$ with rational probabilities, there is an
equivalent hidden-variable model $(\Omega,p)$ which satisfies Weak Determinism
and $\lambda$-Independence.
\end{theorem}

The idea of the proof is to have the hidden variable live on the (discretized)
unit interval, and then to split up the interval according to relevant
probabilities.%
\medskip

\begin{proof}
We give a proof for the case that $\Psi$ is a 4-way product, but the
argument clearly extends. Let $p$ be arbitrary (but with full support) on $\{A,A^{\prime},\ldots
\}\times\{B,B^{\prime},\ldots\}$. \ Fix a pair of settings $(A_{i},B_{j})$
with $q(A_{i},B_{j})>0$. \ For a pair of outcomes $(a_{k},b_{l})$, write the
conditional probability as%
\[
q(a_{k},b_{l}|A_{i},B_{j})=\frac{r_{ijkl}}{s_{ijkl}}%
\]
for some integers $r_{ijkl}\geq0$ and $s_{ijkl}>0$.

Let $\Lambda$ have $N$ points, where%
\[
N=\prod\limits_{i,j,k,l}s_{ijkl}.
\]
We let $p$ be uniform on $N$, and then form the product on $\{A,A^{\prime
},\ldots\}\times\{B,B^{\prime},\ldots\}\times\Lambda$. \ Thus, $\lambda$-Independence
is satisfied.

Still fixing $(A_{i},B_{j})$, we `assign' to $(a_{k},b_{l})$ the number of
points%
\[
r_{ijkl}\times\prod\limits_{\substack{k^{\prime}\neq k \\l^{\prime}\neq l
}}s_{ijk^{\prime}l^{\prime}}\times\prod\limits_{\substack{i^{\prime}\neq i
\\j^{\prime}\neq j}}\prod\limits_{k,l}s_{i'j'kl}.
\]

Formally, we mean that for each of these points,%
\[
p(a_{k},b_{l}|A_{i},B_{j},\lambda)=1.
\]

Thus, Weak Determinism is satisfied. \ Note that, again for fixed
$(A_{i},B_{j})$, the probability of choosing one of these points is%
\[
\frac{r_{ijkl}\times\prod_{\substack{k^{\prime}\neq k \\l^{\prime}\neq l
}}s_{ijk^{\prime}l^{\prime}}\times\prod_{\substack{i^{\prime}\neq i
\\j^{\prime}\neq j}}\prod_{k,l}s_{i'j'kl}}{\prod_{i,j,k,l}s_{ijkl}}%
=\frac{r_{ijkl}}{s_{ijkl}},
\]
and so%
\[
p(a_{k},b_{l}|A_{i},B_{j})=\frac{r_{ijkl}}{s_{ijkl}}=q(a_{k},b_{l}|A_{i}%
,B_{j}),
\]
establishing equivalence.
\end{proof}

\begin{corollary}
Given a hidden-variable model $(\Omega,p)$ with rational probabilities, there
is an equivalent hidden-variable model $(\Omega,\overline{p})$ which satisfies
Weak Determinism and $\lambda$-Independence.
\end{corollary}

\begin{proof}
As for Corollary 3.1, using Theorem 3.2 in place of Theorem 3.1.
\end{proof}%

\medskip
We could handle irrational probabilities in Theorem 3.2, if we allowed an
infinite $\Lambda$. \ Set $\Lambda=[0,1]$ and again take $p$ to be uniform on
$\Lambda$. \ Fix again a pair of settings $(A,B)$ with $q(A,B)>0$. \ Write the
support of $q(\cdot,\cdot|A,B)$ as $\{(a_{1},b_{1}),\ldots,(a_{n},b_{n}%
)\}$.\ \ Partition $\Lambda$ as in Figure 3.3. \ On element $\Lambda^{k}$ of
the partition, we set $p(a_{k},b_{k}|A,B,\lambda)=1$. \ Conceptually, $\lambda$-Independence and
Weak Determinism follow as before. \ (But to make this formal,
we would need to extend these definitions to infinite $\Lambda$.)%

\begin{center}
\includegraphics[
natheight=7.499600in,
natwidth=9.999800in,
height=1.079in,
width=4.9159in
]
{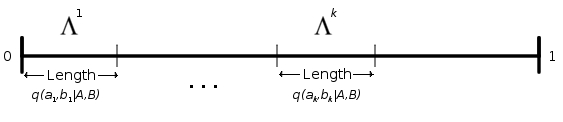}
\\
Figure 3.3
\end{center}

Finally in this section, we record the obvious fact:

\begin{remark}
Given an empirical model $(\Psi,q)$, there is an equivalent hidden-variable
model $(\Omega,p)$ which satisfies Single-Valuedness.
\end{remark}

\begin{proof}
Let $\Lambda=\{\lambda\}$, and, for $a,b,A,B$, let $p(a,b,A,B,\lambda
)=q(a,b,A,B)$.
\end{proof}

\section{EPR}

We have seen that certain types of hidden-variable model are always possible.
\ Next come the no-go theorems, expressed in our framework. \ These show that
no other type of hidden-variable model (among those covered by our six
conditions) is necessarily possible.

Here is the first no-go theorem, due to EPR \cite[1935]{einstein-podolsky-rosen35}, expressed in our framework. \ (Our formulation is
very similar to that in Norsen \cite[2004]{norsen04}.) \ Figure 4.1 adds
crosses to Figure 3.1, in accordance with the EPR result.%

Note that, as with our other statements, EPR as given here is a simple result
in probability theory. \ But the notation we use in the empirical model is
meant to reflect the underlying physical set-up which was of interest to EPR.
\ (More precisely, it reflects Bohm's \cite[1951]{bohm51} reformulation of
EPR.) \ In the physical set-up, there are two entangled particles that are
anti-correlated. \ If Ann measures positive spin, then Bob measures negative
spin, and vice versa. \ There is a 50-50 chance of each pair of outcomes.

\begin{theorem}
[{EPR \cite[1935]{einstein-podolsky-rosen35}}]There is an empirical model
$(\Psi,q)$ for which there is no equivalent hidden-variable model $(\Omega,p)$
which satisfies Single-Valuedness and Outcome Independence.
\end{theorem}

\[%
{\parbox[b]{5.5936in}{\begin{center}
\includegraphics[
trim=0.000000in 0.000000in -0.381536in 0.000000in,
natheight=5.010700in,
natwidth=8.083400in,
height=3.3217in,
width=5.5936in
]%
{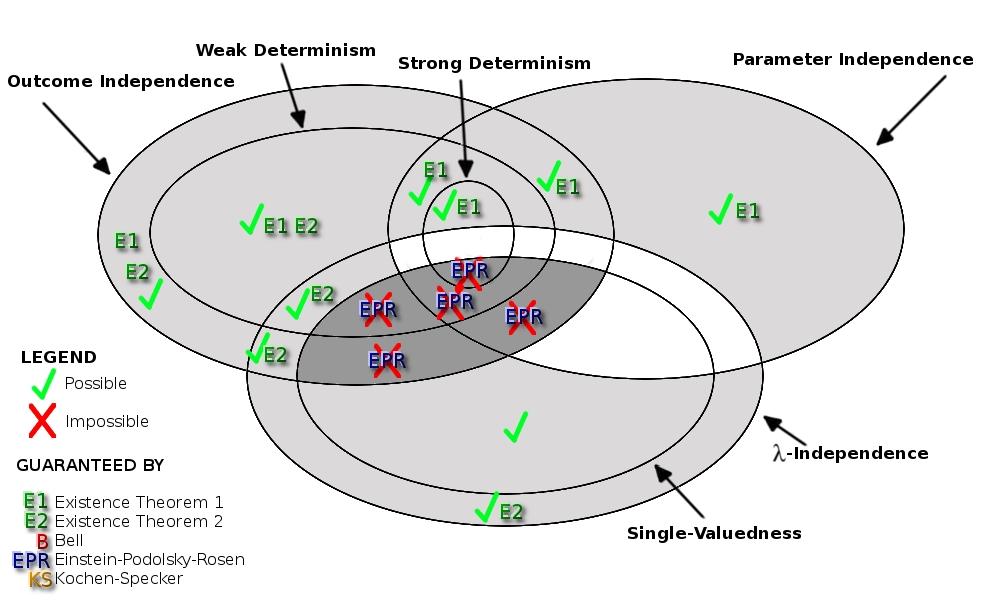}%
\\
Figure 4.1
\end{center}}}%
\]

\begin{proof}
We let%
\[
\Psi=\{+_{a},-_{a}\}\times\{+_{b},-_{b}\}\times\{A\}\times\{B\},
\]
and define $q$ as in Figure 4.2:%

\vspace{.5in}
\begin{center}
\begin{tabular}{l||c|c||}
 &$+_b$&$-_b$\\ \hline \hline&~& \\
$+_a$&0&$\frac{1}{2}$ \\  &~& \\ \hline &~& \\
$-_a$&$\frac{1}{2}$&$0$ \\ &~& \\\hline\hline
\end{tabular} $\qquad q(\cdot, \cdot |A,B)$
\end{center}
\begin{center} Figure 4.2 ~~~~~~~~~~~~\end{center}
\vspace{.2in}

Now suppose, contra hypothesis, there is an equivalent hidden-variable model
$(\Omega,p)$ satisfying Single-Valuedness and Outcome Independence. \ Let
$\Lambda=\{\lambda\}$. \ Then we must have%
\[
p(+_{a},-_{b}|A,B,\lambda)=p(-_{a},+_{b}|A,B,\lambda)=\frac{1}{2},
\]
from which%
\[
p(+_{a}|A,B,\lambda)=p(+_{a},+_{b}|A,B,\lambda)+p(+_{a},-_{b}|A,B,\lambda
)=0+\frac{1}{2},
\]
and%
\[
p(+_{a}|A,B,-_{b},\lambda)=\frac{p(+_{a},-_{b}|A,B,\lambda)}{p(-_{b}%
|A,B,\lambda)}=\frac{\frac{1}{2}}{\frac{1}{2}}=1,
\]
contradicting Outcome Independence.
\end{proof}%

\medskip
The conditions of EPR are tight. \ By Remark 3.1, we cannot drop Outcome
Independence. \ By Theorem 3.1 or 3.2, we cannot drop Single-Valuedness.
\ Here is a specific construction--for the EPR empirical model--of an
equivalent hidden-variable model satisfying Strong Determinism (so, certainly
Outcome Independence) and even $\lambda$-Independence. \ Let $\Lambda=\{\lambda
^{1},\lambda^{2}\}$, and set $p(\lambda^{1})=p(\lambda^{2})=\frac{1}{2}$ and%
\begin{align*}
p(+_{a},-_{b}|A,B,\lambda^{1})  &  =1,\\
p(-_{a},+_{b}|A,B,\lambda^{2})  &  =1.
\end{align*}
Using $p(A,B)=1$, we see that the stated conditions hold.


At the level presented here, the EPR argument doesn't need any quantum
effects. \ It could be realized entirely classically. \ Von Neumann
\cite[1936]{vonneumann36} gave a nice example of classical action at a distance:

\begin{quote}
Let $S_{1}$ and $S_{2}$ be two boxes. \ One knows that $1,000,000$ years ago
either a white ball had been put into each or a black ball had been placed
into each but one does not know which color the balls were. \ Subsequently one
of the boxes ($S_{1}$) was buried on Earth, the other ($S_{2}$) on Sirius
\ldots. \ Now one digs $S_{1}$ on Earth out, opens it and sees: the ball is
white. \ This action on Earth changes instantaneously the $S_{2}$ statistic on
Sirius \ldots.
\end{quote}

In the QM context, EPR's conclusion was that the theory of QM needed to be \textquotedblleft
completed.\textquotedblright\ \ This leads to the question of whether a
construction like the one we just gave is always possible. \ This then leads
to Bell's Theorem.

\section{Bell}

Bell's Theorem adds crosses to Figure 4.1, as in Figure 5.1.%
\[%
{\parbox[b]{5.5936in}{\begin{center}
\includegraphics[
trim=0.000000in 0.000000in -0.381536in 0.000000in,
natheight=5.010700in,
natwidth=8.083400in,
height=3.3217in,
width=5.5936in
]%
{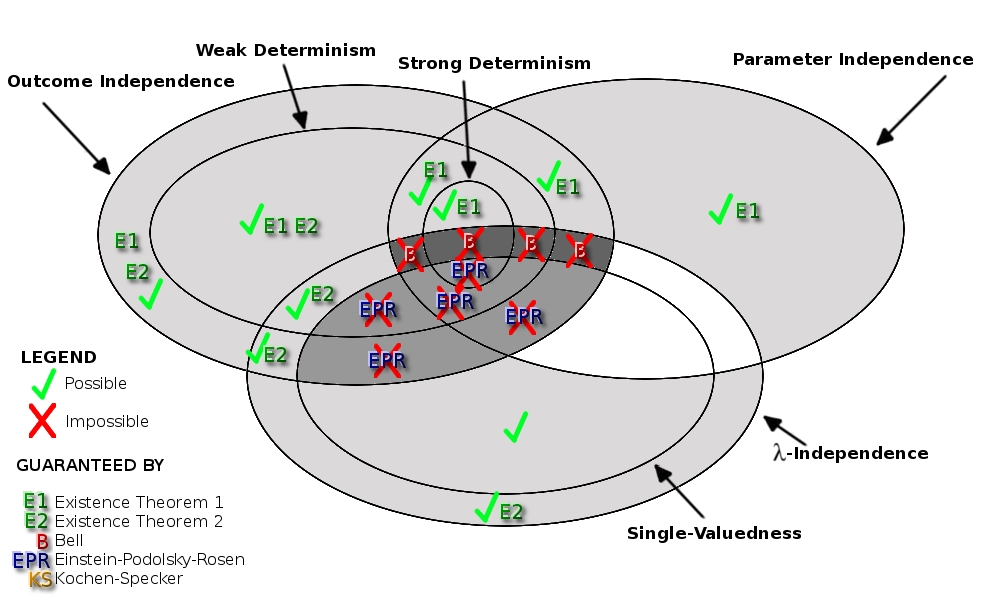}%
\\
Figure 5.1
\end{center}}}%
\]

Once more, our formulation is in probability terms alone. \ In the Bell
experiment, Ann (resp.%
\thinspace
Bob) can make measurements of spin on her (resp.%
\thinspace
his) entangled particle in three directions. \ For each
measurement, the only possible outcome is positive or negative spin. \ If the
measurements are made in the same direction, the results will be
anti-correlated (Figure 5.2). \ Figure 5.3 gives the probabilities of the
different outcomes of the measurements, when these are made in different
directions. \ The probabilities in Figure 5.3 are essentially quantum-mechanical.

\begin{theorem}
[{Bell \cite[1964]{bell64}}]There is an empirical model $(\Psi,q)$ for which
there is no equivalent hidden-variable model $(\Omega,p)$ which satisfies $\lambda$-Independence, Parameter Independence, and Outcome Independence.
\end{theorem}

Another phrasing (using Proposition 2.1):\ There is no equivalent
hidden-variable model which satisfies $\lambda$-Independence and Locality.%
\medskip

\begin{proof}
We let%
\[
\Psi=\{+_{a},-_{a}\}\times\{+_{b},-_{b}\}\times\{A_{1},A_{2},A_{3}%
\}\times\{B_{1},B_{2},B_{3}\},
\]
and define $q$ as in Figures 5.2 and 5.3, with $q(A_{i},B_{j})=\frac{1}{9}$
for all $i,j$.%

\vspace{.5in}
\begin{center}
\begin{tabular}{l||c|c||}
 &$+_b$&$-_b$\\ \hline \hline&~& \\
$+_a$&0&$\frac{1}{2}$ \\  &~& \\ \hline &~& \\
$-_a$&$\frac{1}{2}$&$0$ \\ &~& \\\hline\hline
\end{tabular} $\qquad \qquad ~~~~~~q(\cdot, \cdot |A_i,B_i)$
\end{center}
\begin{center} Figure 5.2~~~~~~~~~~~~~~~~~~~~~~~~~~~ \end{center}

\vspace{.2in}

\begin{center}
\begin{tabular}{l||c|c||}
 &$+_b$&$-_b$\\ \hline \hline&~& \\
$+_a$&$\frac{3}{8}$&$\frac{1}{8}$ \\  &~& \\ \hline &~& \\
$-_a$&$\frac{1}{8}$&$\frac{3}{8}$\\ &~& \\\hline\hline
\end{tabular} $\qquad q(\cdot, \cdot |A_i,B_j)$ for $j\neq i$
\end{center}
\begin{center} Figure 5.3~~~~~~~~~~~~~~~~~~~~~~~~~ \end{center}
\vspace{.2in}



Now suppose, contra hypothesis, there is an equivalent hidden-variable model
$(\Omega,p)$ satisfying $\lambda$-Independence, Parameter Independence, and Outcome Independence.

Fix an $i$. \ By assumption, $p(A_{i},B_{i})>0$, since $q(A_{i},B_{i})>0$.
\ Using Figure 5.2, we have%
\begin{multline*}
0=q(+_{a},+_{b}|A_{i},B_{i})=\sum\limits_{\{\lambda:p(A_{i},B_{i}%
,\lambda)>0\}}p(+_{a},+_{b}|A_{i},B_{i},\lambda)p(\lambda|A_{i},B_{i})=\\
\sum\limits_{\{\lambda:p(A_{i},B_{i},\lambda)>0\}}p(+_{a},+_{b}|A_{i}%
,B_{i},\lambda)p(\lambda)=\\
\sum\limits_{\{\lambda:p(A_{i},B_{i},\lambda)>0\}}p(+_{a}|A_{i},\lambda
)p(+_{b}|B_{i},\lambda)p(\lambda),
\end{multline*}
where the second line uses $\lambda$-Independence and the third line uses Parameter
Independence, Outcome Independence, and Proposition 2.1. \ Using
$p(A_{i},B_{i})>0$ and $\lambda$-Independence again, we have $p(A_{i},B_{i},\lambda)>0
$ if and only if $p(\lambda)>0$. \ Let $M=\{\lambda:p(\lambda)>0\}$. \ Then,%
\begin{equation}
p(+_{a}|A_{i},\lambda)\times p(+_{b}|B_{i},\lambda)=0 \tag{5.1}%
\end{equation}
whenever $\lambda\in M$.

A similar argument using $q(-_{a},-_{b}|A_{i},B_{i})=0$ establishes%
\begin{equation}
p(-_{a}|A_{i},\lambda)\times p(-_{b}|B_{i},\lambda)=0 \tag{5.2}%
\end{equation}
whenever $\lambda\in M$.

Using (5.1) and (5.2), we see that for each $i$, there are disjoint sets
$K_{i},L_{i}\subseteq\Lambda$, with $K_{i}\cup L_{i}=M$, such that%
\begin{equation}%
\begin{array}
[c]{ll}%
p(+_{a}|A_{i},\lambda)=1\text{ and }p(-_{b}|B_{i},\lambda)=1 & \text{when
}\lambda\in K_{i},\\
p(-_{a}|B_{i},\lambda)=1\text{ and }p(+_{b}|B_{i},\lambda)=1 & \text{when
}\lambda\in L_{i}.
\end{array}
\tag{5.3}%
\end{equation}

Similar to above, observe that%
\begin{equation}
q(+_{a},+_{b}|A_{i},B_{j})=\sum\limits_{M}p(+_{a}|A_{i},\lambda)p(+_{b}%
|B_{j},\lambda)p(\lambda). \tag{5.4}%
\end{equation}

Using (5.3) (for $i$ and $j$) in (5.4) we get%
\[
q(+_{a},+_{b}|A_{i},B_{j})=p(K_{i}\cap L_{j}).
\]

A parallel argument yields%
\[
q(-_{a},-_{b}|A_{i},B_{j})=p(L_{i}\cap K_{j}).
\]

Now use Figure 5.3 to get%
\begin{equation}
p(K_{i}\cap L_{j})+p(L_{i}\cap K_{j})=\frac{3}{4} \tag{5.5}%
\end{equation}
whenever $i\neq j$.%
\begin{center}
\includegraphics[
natheight=7.499600in,
natwidth=9.999800in,
height=2.738in,
width=2.8266in
]
{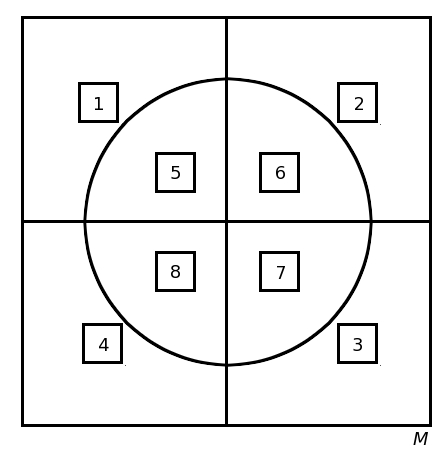}
\\
Figure 5.4
\end{center}

Refer to Figure 5.4 (similar to figures in d'Espagnat \cite[1979]%
{despagnat79}), and let%
\begin{align*}
K_{1}  &  =\fbox{1}\cup\fbox{4}\cup\fbox{5}\cup\fbox{8},\\
L_{1}  &  =\fbox{2}\cup\fbox{3}\cup\fbox{6}\cup\fbox{7},\\
K_{2}  &  =\fbox{1}\cup\fbox{2}\cup\fbox{5}\cup\fbox{6},\\
L_{2}  &  =\fbox{3}\cup\fbox{4}\cup\fbox{7}\cup\fbox{8},\\
K_{3}  &  =\fbox{1}\cup\fbox{2}\cup\fbox{3}\cup\fbox{4},\\
L_{3}  &  =\fbox{5}\cup\fbox{6}\cup\fbox{7}\cup\fbox{8}.
\end{align*}

Now (5.5) for $(i,j)=(1,2)$, $(2,3)$, and $(3,1)$ respectively, yields%
\begin{align*}
p(K_{1}\cap L_{2})+p(L_{1}\cap K_{2})  &  =\frac{3}{4},\\
p(K_{2}\cap L_{3})+p(L_{2}\cap K_{3})  &  =\frac{3}{4},\\
p(K_{3}\cap L_{1})+p(L_{3}\cap K_{1})  &  =\frac{3}{4},
\end{align*}
or%
\begin{align}
p(\fbox{4})+p(\fbox{8})+p(\fbox{2})+p(\fbox{6})  &  =\frac{3}{4},\tag{5.6}\\
p(\fbox{5})+p(\fbox{6})+p(\fbox{3})+p(\fbox{4})  &  =\frac{3}{4},\tag{5.7}\\
p(\fbox{2})+p(\fbox{3})+p(\fbox{5})+p(\fbox{8})  &  =\frac{3}{4}. \tag{5.8}%
\end{align}

Adding (5.6)-(5.8) gives%
\[
2\times(p(\fbox{2})+p(\fbox{3})+p(\fbox{4})+p(\fbox{5})+p(\fbox{6}%
)+p(\fbox{8}))=\frac{9}{4},
\]
or%
\[
p(\fbox{2})+p(\fbox{3})+p(\fbox{4})+p(\fbox{5})+p(\fbox{6}+p(\fbox{8}%
)=\frac{9}{8},
\]
which is impossible.
\end{proof}%

\medskip
Can we drop any of the conditions of Bell's Theorem? \ By Theorem 3.1, we
cannot drop $\lambda$-Independence. \ By Theorem 3.2, we cannot drop Parameter Independence.

For the Bell empirical model, we also cannot drop Outcome Independence. \ To
see this, let $\Lambda=\{\lambda\}$. \ Then $\lambda$-Independence is satisfied.
\ Define $p$ on $\Psi\times\{\lambda\}$ from $q$ on $\Psi$, as in Remark 3.1.
\ We then have%
\begin{align*}
p(+_{a}|A_{i},B_{i},\lambda)  &  =q(+_{a}|A_{i},B_{i})=\frac{1}{2}%
=q(+_{a}|A_{i},B_{j})=p(+_{a}|A_{i},B_{j},\lambda),\\
p(+_{b}|A_{i},B_{i},\lambda)  &  =q(+_{b}|A_{i},B_{i})=\frac{1}{2}%
=q(+_{b}|A_{j},B_{i})=p(+_{b}|A_{j},B_{i},\lambda),
\end{align*}
so that Parameter Independence is satisfied. \ Of course, Outcome Independence
fails, as it must. \ For example:%
\[
p(+_{a}|-_{b},A_{i},B_{i},\lambda)=1\neq\frac{1}{2}=p(+_{a}|A_{i}%
,B_{i},\lambda).
\]

By contrast, the Kochen-Specker Theorem produces an impossibility even without
Outcome Independence.

\section{Kochen-Specker}

The Kochen-Specker \cite[1967]{kochen-specker67} no-go result adds crosses to
Figure 5.1, to give a complete picture as in Figure 1.1.%

At the physical level, the Kochen-Specker experiment differs from
those in the past two sections in considering measurements on only
one particle. There are many presentations of Kochen-Specker, of
course. We follow Cabello, Estebaranz, and Garc\`{\i}a-Alcaine
\cite[1996]{cabello-estebaranz-garciaalcaine96}, a simple treatment
which results in the $4 \times 9$ array of Table 6.1 (
also presented in Held \cite[2000]{held00}). For various
tuples of four orthogonal directions in 4-space (from a total of 18
directions), we ask whether or not the particle has spin in each of
these directions. \ In each case, the answer will be that we get
three directions without spin and only one direction with spin.


\vspace{.2in}
\begin{center}
\begin{tabular}{||l||c|c|c|c|c|c|c|c|c||}
\hline \hline
$A$&$E_{1 }$&$E_{1 }$&$E_{8 }$&$E_{ 8}$&$E_{ 2}$&$E_{9 }$&$E_{16}$&$E_{16}$&$E_{17}$ \\ \hline
$B$&$E_{2 }$&$E_{5 }$&$E_{9 }$&$E_{11}$&$E_{5 }$&$E_{11}$&$E_{17}$&$E_{18}$&$E_{18}$ \\ \hline
$C$&$E_{3 }$&$E_{6 }$&$E_{3 }$&$E_{7 }$&$E_{13}$&$E_{14}$&$E_{4 }$&$E_{6 }$&$E_{13}$ \\ \hline
$D$&$E_{4 }$&$E_{7 }$&$E_{10}$&$E_{12}$&$E_{14}$&$E_{15}$&$E_{10}$&$E_{12}$&$E_{15}$ \\
\hline\hline
\end{tabular}
\end{center}
\begin{center}Table 6.1 \end{center}

\vspace{.2in}

To state Kochen-Specker in our probabilistic framework, we will need to adapt
the concept of exchangeability from probability theory (de Finetti
\cite[1937]{definetti37}, \cite[1972]{definetti72}). \ To give our definition,
we consider the special case where the spaces of possible measurements are all
the same, as are the spaces of possible outcomes:%
\[
\{A,\ldots\}=\{B,\ldots\}=\cdots=\{X_{1},X_{2},\ldots,X_{m}\},
\]%
\[
\{a,\ldots\}=\{b,\ldots\}=\cdots=\{x_{1},x_{2},\ldots,x_{n}\},
\]
for integers $m,n$. \ We will consider a permutation map $\pi$:%
\begin{align*}
(A,B,\ldots)  &  \mapsto(\pi(A),\pi(B),\ldots),\\
(a,b,\ldots)  &  \mapsto(\pi(a),\pi(b),\ldots).
\end{align*}
Note that we use $\pi$ twice (despite the different domains), because we want
to consider the same permutation on the two sequences.

\begin{definition}
An empirical model $(\Psi,q)$ satisfies \textbf{Exchangeability }if for any
indices $i_{1},i_{2},\ldots\in\{1,2,\ldots,m\}$ and $j_{1},j_{2},\ldots
\in\{1,2,\ldots,n\}$,%
\[
q(A=X_{i_{1}},B=X_{i_{2}},\ldots)>0\text{ if and only if }q(\pi(A)=X_{i_{1}%
},\pi(B)=X_{i_{2}},\ldots)>0,
\]
for any permutation $\pi$, and when both are non-zero,%
\begin{multline*}
q(a=x_{j_{1}},b=x_{j_{2}},\ldots|A=X_{i_{1}},B=X_{i_{2}},\ldots)=\\
q(\pi(a)=x_{j_{1}},\pi(b)=x_{j_{2}},\ldots|\pi(A)=X_{i_{1}},\pi(B)=X_{i_{2}%
},\ldots).
\end{multline*}

\end{definition}

In words, the requirement is that if we swap any number of measurements, then,
as long as we swap the outcomes in the same way, the overall probability is
unchanged. \ Thus, let $q$ be the probability that Ann gets the outcome
$x_{j_{1}}$ and Bob gets the outcome $x_{j_{2}}$, if Ann performs measurement
$X_{i_{1}}$ on her particle and Bob performs measurement $X_{i_{2}}$ on his
particle. \ Let $q^{\prime}$ be the probability that Ann gets the outcome
$x_{j_{2}}$ and Bob gets the outcome $x_{j_{1}}$, if Ann performs measurement
$X_{i_{2}}$ on her particle and Bob performs measurement $X_{i_{1}}$ on his
particle. Exchangeability says that $q^{\prime}=q$. \ Likewise, for several
measurements on a single particle. \ This is similar to exchangeability \`{a}
la de Finetti, though with a conditioning component.

Exchangeability might come from physical arguments. \ For example, the Bell
model (Figures 5.2 and 5.3) satisfies Exchangeability. \ (This reflects the
underlying physical fact that only the angle between the two measurements matters.)

\begin{theorem}
[{Kochen-Specker \cite[1967]{kochen-specker67}}]There is an empirical model
$(\Psi,q)$ for which there is no equivalent hidden-variable model that
satisfies $\lambda$-Independence and Parameter Independence.
\end{theorem}

Kochen-Specker demonstrated the existence of a QM model that fails
Non-Contextuality:\ Whether or not their particle has spin in a certain
direction is dependent on which other directions are also measured. \ The
property of spin for such a particle does not stand alone. \ As the proof
makes clear, Theorem 6.1 is really a corollary to their result.%
\medskip

\begin{proof}
Consider an empirical model where%
\[
\{A,\ldots\}=\{B,\ldots\}=\{C,\ldots\}=\{D,\ldots\}=\{E_{1},E_{2}%
,\ldots,E_{18}\},
\]%
\[
\{a,\ldots\}=\{b,\ldots\}=\{c,\ldots\}=\{d,\ldots\}=\{0,1\}.
\]
Exchangeability is assumed to hold, and $q$ assigns positive probability to
each of the nine tuples of measurement settings.

Finally, for any column, the empirical model has the property that precisely
one of the following holds:%
\begin{align}
q(1,0,0,0|E_{i_{1}},E_{i_{2}},E_{i_{3}},E_{i_{4}})  &  =1,\tag{6.1}\\
q(0,1,0,0|E_{i_{1}},E_{i_{2}},E_{i_{3}},E_{i_{4}})  &  =1,\tag{6.2}\\
q(0,0,1,0|E_{i_{1}},E_{i_{2}},E_{i_{3}},E_{i_{4}})  &  =1,\tag{6.3}\\
q(0,0,0,1|E_{i_{1}},E_{i_{2}},E_{i_{3}},E_{i_{4}})  &  =1. \tag{6.4}%
\end{align}

Now suppose, contra hypothesis, that there is an equivalent hidden-variable
model satisfying $\lambda$-Independence and Parameter Independence. \ By Proposition
2.2, the above empirical model then satisfies Non-Contextuality.

Next, take, say, the first column. \ If%
\begin{equation}
q(0,1,0,0|E_{1},E_{2},E_{3},E_{4})=1, \tag{6.5}%
\end{equation}
then certainly%
\[
q(b=1|E_{1},E_{2},E_{3},E_{4})=1.
\]

Since $(E_{2},E_{5},E_{13},E_{14})$ is non-null, so is $(E_{5},E_{2}%
,E_{13},E_{14})$, by Exchangeability. \ Using Non-Contextuality, we therefore
have%
\[
q(b=1|E_{5},E_{2},E_{13},E_{14})=1,
\]
from which, by Exchangeability again,%
\[
q(a=1|E_{2},E_{5},E_{13},E_{14})=1.
\]

Now use (6.1)-(6.4) to get%
\begin{equation}
q(1,0,0,0|E_{2},E_{5},E_{13},E_{14})=1, \tag{6.6}%
\end{equation}
which tells us about the fifth column.

We therefore get a coloring problem: We try to color precisely one entry in
each column--corresponding to the measurement that yields a $1$. \ For
example, suppose we color the entry $E_{2}$ in the first column--corresponding
to (6.5).\ \ Then (6.6) tells us that we must color the entry $E_{2}$ in the
fifth column. \ However, this is impossible. \ Each $E_{i}$ appears an even
number of times in Table 6.1, and there is an odd number of columns. \ Thus,
the table cannot be colored.
\end{proof}

\section{Other No-Go Theorems}%

There are many important papers on the no-go question not touched
upon here. These include Fine \cite[1982]{fine82a},
\cite[1982]{fine82b}, Greenberger, Horne, and Zeilinger \cite[1989]{GHZ89}, Malley and Fine \cite[2005]{malley-fine05},
Mermin \cite[1990]{mermin90}, \cite[1993]{mermin93}, Peres \cite[1990]{peres90},
\cite[1991]{peres91}, and Szabo and Fine
\cite[2002]{szabo-fine02}. Again, our purpose is not to survey the
literature. Rather, it is to give a complete picture of Figure 1.1
and all its 21 regions. As Figure 1.1 shows, just the three basic
no-go theorems are needed for the six properties that we present.


The absence of Gleason's Theorem (\cite[1957]{gleason57})
from our paper is a consequence of our choice not to impose any structure
on our spaces (refer back to Footnote 4). In particular, we do not work in
Hilbert space. Of course, Gleason's Theorem immediately
implies the existence of the Kochen-Specker QM
model (which we used in our Theorem 6.1).%

The recent no-go theorem of Conway and Kochen \cite[2006]{conway-kochen06}
generalizes Kochen-Specker by relaxing Parameter
Independence. \ Consider a two-particle system. \ The requirement is that,
conditional on the value of the hidden variable, the outcome of any particular
measurement Ann makes on her particle may depend (probabilistically) on the
other measurements she makes but not on the measurements Bob makes on his
particle. \ We could accommodate this result by adding a seventh property--a generalized
parameter independence--to our six, but refrain from pursuing this extension here.


Finally, we note the connection to Bohmian mechanics (Bohm
\cite[1952]{bohm52}). \ D\"{u}rr-Goldstein-Zangh\`{\i} \cite[2004,
p.993]{durr-goldstein-zanghi04} explain: \textquotedblleft In Bohmian
mechanics the result obtained at one place at any given time will in fact
depend upon the choice of measurement simultaneously performed at the other
place.\textquotedblright\ \ Indeed, Theorem 3.2 says that provided one is
prepared to give up Parameter Independence, one can reproduce any empirical
model--under $\lambda$-Independence and Weak Determinism. \ Theorem 3.1 says that if one is prepared to give up
$\lambda$-Independence, one can get even Strong Determinism. \ In a sense, then, these results `predict' the
possibility of Bohmian mechanics--though not its specific content, of course.%

\newpage

\end{document}